\documentclass[11pt]{article}              

\usepackage[margin=20mm]{geometry}
\usepackage{graphicx}
\usepackage{colortbl}
\usepackage{fancyhdr}
\usepackage{amssymb,amsfonts,amsmath,amsthm}
\usepackage{amstext}

\usepackage{lineno,hyperref}

\usepackage{epsfig}
\usepackage{cmap}
\usepackage{mathrsfs}
\usepackage{amssymb}
\usepackage{amsfonts}
\usepackage{algorithm}
\usepackage{epsfig}
\usepackage{caption}
\usepackage{subcaption}

\newtheorem{remark}{Remark}
\newtheorem{definition}{Definition}
\newtheorem{theorem}{Theorem}

\usepackage{authblk}
	
\graphicspath{{./pics/}}

\usepackage{url}

\newcommand{\ton}{\underset{n \to +\infty}{\longrightarrow}}
\newcommand{\tod}{\overset{\mathfrak{D}}{\underset{n \to +\infty}{\longrightarrow}}}

\newcommand{\Sb}{\mathbf{S}}

\newcommand{\xb}{\mathbf{x}}
\newcommand{\cb}{\mathbf{c}}
\newcommand{\XX}{\mathscr{X}}
\newcommand{\DD}{\mathscr{H}}
\newcommand{\rr}{\mathscr{L}}
\newcommand{\rrst}{\mathscr{L}_*}

\newcommand{\NN}{\mathbb{N}}
\newcommand{\Zp}{\NN}
\newcommand{\ee}{\mathbb{E}}
\newcommand{\al}{\boldsymbol\alpha}
\newcommand{\nub}{\boldsymbol\nu}
\newcommand{\be}{\boldsymbol\beta}
\newcommand{\ub}{\mathbf{u}}
\newcommand{\vb}{\mathbf{v}}
\newcommand{\wb}{\mathbf{w}}
\newcommand{\tb}{\mathbf{t}}
\newcommand{\R}{\mathbb{R}} 
\newcommand{\phx}{\boldsymbol{\Psi}(\mathbf{x})} 
\newcommand{\phTx}{\boldsymbol{\Psi}^{T}(\mathbf{x})} 
\newcommand{\E}{\mathbb{E}}
\newcommand{\V}{\mathbb{V}}
\newcommand{\Nd}{\mathcal{N}}
\newcommand{\dU}{\mathcal{U}}

\begin{document}
\title{Efficient design of experiments for sensitivity analysis based on polynomial chaos expansions} 

\author[1,2,3]{E. Burnaev} \author[2]{I. Panin} \author[4]{B. Sudret} 

\affil[1]{Skolkovo Institute of Science and Technology, Building 3,
  Nobel st., Moscow 143026, Russia} 
\affil[2]{Kharkevich Institute for Information Transmission Problems,
  Bolshoy Karetny per. 19, Moscow 127994, Russia}
\affil[3]{National Research University Higher School of Economics\\
  Myasnitskaya st. 20, Moscow 109028, Russia}
\affil[4]{Chair of Risk, Safety and Uncertainty Quantification, ETH
  Zurich, Stefano-Franscini-Platz 5, 8093 Zurich, Switzerland}
\date{}
\maketitle

\abstract{ Global sensitivity analysis aims at quantifying respective
  effects of input random variables (or combinations thereof) onto
  variance of a physical or mathematical model response. Among the
  abundant literature on sensitivity measures, Sobol' indices have
  received much attention since they provide accurate information for
  most of models. We consider a problem of experimental design points
  selection for Sobol' indices estimation. Based on the concept of
  $D$-optimality, we propose a method for constructing an adaptive
  design of experiments, effective for calculation of Sobol' indices
  based on Polynomial Chaos Expansions. We provide a set of applications
  that demonstrate the efficiency of the proposed approach. \\[1em] 
  
  }

\maketitle


\section{Introduction}
\label{sec: Intro}

Computational models play important role in different areas of human activity (see \cite{Beven2000,Dayan2001,Burn1}). Over the past decades, computational models have become more complex, and there is an increasing need for special methods for their analysis. \emph{Sensitivity analysis} is an important tool for investigation of computational models.

\emph{Sensitivity analysis} tries to find how different model input parameters influence the model output, what are the most influential parameters and how to evaluate such effects quantitatively (see \cite{LocalAndGlobal}). Sensitivity analysis allows to better understand behavior of computational models. Particularly, it allows us to separate all input parameters into {\it important (significant)}, {\it relatively important} and {\it unimportant (nonsignificant)} ones. Important parameters, {\it i.e.} parameters whose variability has a strong effect on the model output, need to be controlled more accurately. Complex computational models often suffer from over-parameterization. By excluding unimportant parameters, we can potentially improve model quality, reduce parametrization (which is of great interest in the field of meta-modeling) and computational costs \cite{ElstatHastie2009}.


Sensitivity analysis includes a wide range of metrics and techniques: {\it e.g.} the Morris method \cite{Morris1991}, linear regression-based methods \cite{Iooss2015}, variance-based methods \cite{Saltelli2008}. Among others, \emph{Sobol' (sensitivity) indices} are a common metric to evaluate the influence of model parameters \cite{Sobol93}. Sobol' indices quantify which portions of the output variance are explained by different input parameters and combinations thereof. This method is especially useful for the case of nonlinear computational models \cite{Saltelli2010}.

There are {two main approaches} to evaluate  Sobol' indices. {\it Monte Carlo approach} (Monte Carlo simulations, FAST~\cite{FASTCukier1978}, SPF scheme~\cite{SPFSobol2001} and others) is relatively robust (see \cite{Yang2011}), but requires large number of model runs, typically in the order of $10^4$ for an accurate estimation of each index. Thus, it is impractical for a number of industrial applications, where each model evaluation is computationally costly.

{\it Metamodeling approaches} for Sobol' indices estimation allow one to reduce the required number of model runs \cite{Iooss2015,Sudret15}. Following this approach, we replace the original computational model by an approximating {\it metamodel} (also known as {\it surrogate model} or {\it response surface}) which is computationally efficient and has some clear internal structure \cite{GTApprox}. The approach consists of the following general steps: selection of  {\it the design of experiments (DoE)} and generation of {\it the training sample}, construction of the metamodel based on the training sample, including its accuracy assessment and evaluation of Sobol' indices (or any other measure) using the constructed metamodel. Note that the evaluation of indices may be either based on a known internal structure of the metamodel or via Monte Carlo simulations based on the metamodel itself.

In general, the metamodeling approach is more computationally efficient than an original Monte Carlo approach, since the cost (in terms of the number of runs of the costly computational model) reduces to that of the training set (usually containing results from a few dozens to a few hundreds model runs). However, this approach can be nonrobust and its accuracy is more difficult to analyze. Indeed, although procedures like {\it cross-validation} \cite{ElstatHastie2009,Stone1974} allow to estimate quality of metamodels, the accuracy of complex statistics ({\it e.g.} Sobol' indices), derived from metamodels, has a complicated dependency on the metamodels structure and quality (see {\it e.g.} confidence intervals for Sobol' indices estimates \cite{Marrel2009} in case of Gaussian Process metamodel \cite{Burn3,Burnaev13,BurnaevGP,Belyaev15_SLDS,Belyaev15_SLDS_v2} and bootstrap-based confidence intervals in case of polynomial chaos expansions \cite{CIbootstrapSobol}).

In this paper, we consider a problem of a DoE construction in case of a particular metamodeling approach: how to select the experimental design for building a polynomial chaos expansion for further evaluation of Sobol' indices, that is effective in terms of the number of computational model runs? 

Space-filling designs are commonly used for sensitivity analysis. Methods like Monte Carlo sampling, {\it Latin Hypercube Sampling (LHS)} \cite{McKay1979} or sampling in FAST method \cite{FASTCukier1978} try to fill ``uniformly'' the input parameters space with {\it design points} ({\it point}s are some realizations of parameters values). These sampling methods are {\it model free}, as they make no assumptions on the computational model.

In order to speed up the convergence of indices estimates, we assume that the computational model is close to its approximating metamodel and exploit knowledge of the metamodel structure. In this paper, we consider {\it Polynomial Chaos Expansions (PCE)} that is commonly used in engineering and other applications~\cite{Ghiocel2002}. PCE approximation is based on a series of polynomials (Hermite, Legendre, Laguerre etc.) that are orthogonal w.r.t. the probability distributions of corresponding input parameters of the computational model. It allows to calculate Sobol' indices analytically from the expansion coefficients \cite{Sudret,Blatman2010}. 

In this paper, we address the problem of design of experiments construction for evaluating Sobol' indices from a PCE metamodel. Based on asymptotic considerations, we propose an adaptive algorithm for design construction and test it on a set of applied problems. Note that in \cite{BurnaevPanin15}, we investigated the adaptive design algorithm for the case of a quadratic metamodel (see also \cite{Burn2}). In this paper, we extend these results for the case of a generalized PCE metamodel and provide more examples, including real industrial applications.

The paper is organized as follows: in Section \ref{sec: SobolPCE}, we review the definition of sensitivity indices and describe their estimation based on a PCE metamodel. In Section \ref{sec: AsymptoticSobolPCE}, asymptotic analysis of indices estimates is provided. In Section \ref{sec: PCEOptimalDoE}, we introduce an optimality criterion and propose a procedure for constructing the  experimental design. In Section \ref{sec: SAPCEresults}, we provide experimental results,  applications and benchmark with other methods of design construction.

\section{Sensitivity Indices and PCE Metamodel}
\label{sec: SobolPCE}

\subsection{Sensitivity Indices}

Consider a computational model $y = f(\xb)$, where $\xb = (x_1, \ldots, x_d) \in \XX \subset \R^d$~is a vector of \emph{input variables} (aka \emph{parameters} or \emph{features}), $y \in \R^1$ is an \emph{output variable} and $\XX$ is \emph{a design space}. The model $f(\xb)$ describes behavior of some physical system of interest. 

We consider the model $f(\xb)$ as a black-box: no additional knowledge on its inner structure is assumed. For some \emph{design of experiments} $ X = \{\xb_i \in \XX \}_{i = 1}^n \in \R^{n \times d}$ we can obtain a set of model responses and form \emph{a training sample}
\begin{equation}
\label{eq:training_sample}
L = \{\xb_i, y_i = f(\xb_i)\}_{i = 1}^n \triangleq \{X \in \R^{n\times d}, \;  Y = f(X) \in \R^n \}, 
\end{equation}
which allows us to investigate properties of the computational model.
 
Let us assume that there is a prescribed probability distribution $\DD$ with independent marginal distributions  on the design space $\XX$ ($\DD = \DD_1 \times \ldots \times \DD_d$). This distribution represents the uncertainty and/or variability of the input variables, modelled as a random vector $\vec{X} = \{X_1, \ldots, X_d \}$ with independent components. In these settings, the model output $\vec{Y} = f(\vec{X})$ becomes a stochastic variable.

Assuming that the function $f(\vec{X})$ is square-integrable with respect to the distribution $\DD$ ({\it i.e.} $\E [f^2(\vec{X})] < +\infty$), we have the following unique Sobol' decomposition of $\vec{Y} = f(\vec{X})$ (see \cite{Sobol93}) given by
\[
f(\vec{X}) = f_0 + \sum_{i=1}^{d} f_i(X_i) + \sum_{1\leq i \leq j \leq d} f_{ij}(X_i, X_j) + \ldots + f_{1 \ldots d}(X_1, \ldots, X_d),
\]
which satisfies
\[
\E[f_\ub(\vec{X}_\ub)f_\vb(\vec{X}_\vb)] = 0, \; \text{if} \; \ub \neq \vb, 
\]
where $\ub$ and $\vb$ are index sets: $\ub,\vb \subset \{1, 2, \ldots, d\}$.

Due to orthogonality of the summands, we can decompose variance of the model output:
\[
D = \V[f(\vec{X})] = \sum_{\substack{\ub \subset \{1, \ldots, d\}, \\ \ub \neq \mathbf{0}}  } \V[f_\ub(\vec{X}_\ub)] = \sum_{\substack{\ub \subset \{1, \ldots, d\}, \\ \ub \neq \mathbf{0}}  } D_\ub,
\] 
In this expansion $D_\ub \triangleq \V[f_\ub(\vec{X}_\ub)]$ is the contribution of the summand $f_\ub(\vec{X}_\ub)$ to the output variance, also known as the \emph{partial variance}.

\begin{definition}
     {\it The sensitivity index (Sobol' index)} of the subset $\vec{X}_\ub, \; \ub \subset \{1, \ldots, d\}$ of model input variables is defined as
\[
    S_{\ub} = \frac{D_\ub}{D}.  
\]
\end{definition}

The sensitivity index describes the amount of the total variance explained by uncertainties in the subset $\vec{X}_\ub$ of model input variables.

\begin{remark}
In this paper, we consider only sensitivity indices of type $S_i \triangleq S_{\{i\}}, i=1,\ldots,d$, called {\it first-order} or {\it main effect sensitivity indices}.
\end{remark}

\subsection{Polynomial Chaos Expansions}

Consider a set of multivariate polynomials $\{\Psi_{\al}(\vec{X}), \; \al \in \rr\}$ that consists of polynomials $\Psi_{\al}$ having the form of tensor product
\[
\Psi_{\al}(\vec{X}) = \prod_{i=1}^d \psi_{\alpha_i}^{(i)}(X_i), \; \al = \{ \alpha_i \in \Zp, \; i = 1, \ldots, d\} \in \rr,
\]
where $\psi_{\alpha_i}^{(i)}$ is a univariate polynomial of degree $\alpha_i$ belonging to the $i$-th family ({\it e.g.} Legendre polynomials, Jacobi polynomials, etc.), $\Zp = \{0,1,2,\ldots\}$ is the set of nonnegative integers, $\rr$ is some fixed set of multi-indices $\al$. 

Suppose that univariate polynomials $\{\psi_{\alpha}^{(i)}\}$ are orthogonal w.r.t. $i$-th mar\-gi\-nal of the  probability distribution $\DD$, {\it i.e.} $\E[\psi_{\alpha}^{(i)}(X_i)\psi_{\beta}^{(i)}(X_i)] = 0$ if $\alpha \neq \beta$ for $i = 1, \ldots ,d$. Particularly, Legendre polynomials are orthogonal w.r.t. standard uniform distribution;  Hermite polynomials are orthogonal w.r.t. Gaussian distribution. Due to independence of components of $\vec{X}$, we obtain that multivariate polynomials $\{\Psi_{\al}\}$ are orthogonal w.r.t. the probability distribution $\DD$, {\it i.e.} 
\begin{equation}
\label{normalization}
\E[\Psi_{\al}(\vec{X})\Psi_{\be}(\vec{X})] = 0\,\,\mbox{ if }\,\,\al \neq \be. 
\end{equation}

Provided $\E [f^2(\vec{X})] < +\infty$, the spectral polynomial chaos expansion of $f$ takes the form
\begin{equation}
\label{eq:pce_full}
f(\vec{X}) = \sum_{\al \in \NN^d} c_{\al} \Psi_{\al}(\vec{X}),
\end{equation}
where $\{c_{\al} \}_{\al\in\NN^d}$ are expansion coefficients.

In the sequel we consider a PCE approximation $f_{PC}(\vec{X})$ of the model $f(\vec{X})$ obtained by truncating the infinite series to a finite number of terms:
\begin{equation}
\label{eq:pc_approx}
\hat{\vec{Y}} = f_{PC}(\vec{X}) = \sum_{\al \in \rr} c_{\al} \Psi_{\al}(\vec{X}).
\end{equation}
By enumerating the elements of $\rr$ we also use an alternative form of (\ref{eq:pc_approx}):
\[
\hat{\vec{Y}} = f_{PC}(\vec{X}) = \sum_{\al \in \rr} c_{\al} \Psi_{\al}(\vec{X}) \triangleq \sum_{j = 0}^{P-1}c_j \Psi_j(\vec{X}) = \cb^T \boldsymbol\Psi(\vec{X}), \;\; P \triangleq |\rr|,
\]
where $\cb = (c_0, \ldots, c_{P-1})^T$ is a column vector of coefficients and $\phx\colon \R^d \to \R^P$ is a mapping from the~design space to \emph{the~extended design space} defined as a~column vector function $\phx = \left(\Psi_0(\xb), \ldots, \Psi_{P-1}(\xb)\right)^T$. Note that index $j=0$ corresponds to multi-index $\al = \mathbf{0} = \{0, \ldots, 0\}$, {\it i.e.} 
\[
c_{j=0} \triangleq c_{\al=\mathbf{0}}, \;
\Psi_{j=0} \triangleq \Psi_{\al=\mathbf{0}} = const.
\]

The set of multi-indices $\rr$ is determined by some \emph{truncation scheme}. In this work, we use hyperbolic truncation scheme \cite{BlatmanSudret10b}, which corresponds to
\[
\rr = \{\al \in \NN^d: \|\al\|_q\ \leq p\}, \; \|\al \|_q \triangleq \left( \sum_{i=1}^d \alpha_i^q\right)^{1/q}, 
\]
where $q \in (0, 1]$ is a fixed parameter and $p \in \NN \backslash \{0\} = \{1,2,3,\ldots\}$ is a fixed maximal total degree of polynomials. Note that in case of $q=1$, we have $P = \frac{(d+p)!}{d! p!}$ polynomials in $\rr$ and a smaller $q$ leads to a smaller number of polynomials.

There is a number of strategies for estimating the expansion coefficients $c_{\al}$ in (\ref{eq:pc_approx}). In this paper, {\it the least-square (LS)} minimization method is used \cite{Berveiller2006}. Unlike (\ref{eq:pce_full}), the key idea consists in considering the original model $f(\vec{X})$ as the sum of a truncated PC expansion $f_{PC}(\vec{X})$ and a residual $\varepsilon$, i.e. 
\begin{equation}
\label{eq:PCE_and_noise}
f(\vec{X}) = f_{PC}(\vec{X}) +\varepsilon = \sum_{j = 0}^{P-1}c_j \Psi_j(\vec{X}) + \varepsilon = \cb^T \boldsymbol\Psi(\vec{X}) + \varepsilon,
\end{equation}
where thanks to orthogonality property \eqref{normalization} the residual process $\varepsilon$ can be considered as an i.i.d. noise process with $\E\varepsilon = 0$ and $\V[\varepsilon] = \sigma^2$, such that $\varepsilon = \varepsilon(\vec{X})$ and $\{\Psi_j(\vec{X})\}_{j=0}^{P-1}$ are orthogonal w.r.t. the distribution $\DD$.

The coefficients $\cb$ are obtained by minimizing the mean square residual:
\[
\cb = \arg\min_{\cb \in \R^P}\E\left[\left(f(\vec{X}) - \cb^T \boldsymbol\Psi(\vec{X})\right)^2\right],
\]
which is approximated by using the training sample $L = \{\xb_i, y_i = f(\xb_i)\}_{i = 1}^n$:
\begin{equation}
\label{eq:coefficients_ls_opt_problem}
\hat{\cb}_{LS} = \arg\min_{\cb \in \R^P} \frac{1}{n} \sum_{i=1}^n \left[y_i - \cb^T \boldsymbol\Psi(\xb_i)\right]^2.
\end{equation}

\subsection{PCE post-processing for sensitivity analysis}

Consider some PCE model $f_{PC}(\vec{X}) = \sum_{\al \in \rr} c_{\al} \Psi_{\al}(\vec{X}) = \sum_{j = 0}^{P-1} c_{j} \Psi_{j}(\vec{X}) $. According to \cite{Sudret}, we have an explicit form of Sobol' indices (main effects) for model $f_{PC}(\vec{X})$:
\begin{equation}
\label{eq:sobol_pce_basic}
S_i(\cb) = \frac{ \sum_{\al \in \rr_i} c_{\al}^2 \ee[\Psi_{\al}^2(\vec{X})] }{\sum_{\al \in \rrst} c_{\al}^2 \ee[\Psi_{\al}^2(\vec{X})]}, \; i = 1, \ldots, d,
\end{equation}
where $\rrst \triangleq \rr \backslash \{\mathbf{0}\}$ and $\rr_i \subset \rr$ is the set of multi-indices $\al$ such that only index on the $i$-th position~is nonzero: $\al = \{0, \ldots, \alpha_i, \ldots, 0\}$, $\alpha_i \in \NN$, $\alpha_i > 0$.

Suppose for simplicity that the multivariate polynomials $\{\Psi_{\al}(\vec{X}), \; \al \in \rr\}$ are not only orthogonal but also normalized w.r.t. the distribution~$\DD$:
\[
\ee[\Psi_{\al}(\vec{X})\Psi_{\be}(\vec{X})] = \delta_{\al\be},
\]
where $\delta_{\al\be}$ is the Kronecker symbol, {\it i.e} $\delta_{\al\be} = 1$ if $\al=\be$, otherwise $\delta_{\al\be} = 0$. Then (\ref{eq:sobol_pce_basic}) takes the form
\begin{equation}
\label{eq:sobol_pce}
S_i(\cb) = \frac{ \sum_{\al \in \rr_i} c_{\al}^2}{\sum_{\al \in \rrst} c_{\al}^2 }, \; i = 1, \ldots, d.
\end{equation}

Thus, (\ref{eq:sobol_pce}) provides a simple expression for calculation of Sobol' indices in case of the PCE metamodel. If the original model of interest $f(\vec{X})$ is close to its PCE approximation $f_{PC}(\vec{X})$, then we can use expression (\ref{eq:sobol_pce}) for indices with estimated coefficients (\ref{eq:coefficients_ls_opt_problem}) to approximate Sobol' indices of the original model:
\begin{equation}
\label{eq:sobol_pce_approximation}
\hat{S}_i = S_i(\hat{\cb}) = \frac{ \sum_{\al \in \rr_i} \hat{c}_{\al}^2}{\sum_{\al \in \rrst} \hat{c}_{\al}^2 }, \; i = 1, \ldots, d,
\end{equation}
where $\hat{\cb} \triangleq \hat{\cb}_{LS}$.

\section{Asymptotic Properties}
\label{sec: AsymptoticSobolPCE}

In this section, we consider asymptotic properties of indices estimates in Eq.~(\ref{eq:sobol_pce_approximation}) if the coefficients $\cb$ are estimated by LS approach (\ref{eq:coefficients_ls_opt_problem}).  Let $\hat{\cb}_n$ be LS estimate (\ref{eq:coefficients_ls_opt_problem}) of the true coefficients vector $\cb$ based on the training sample $L = \{\xb_i, y_i = f(\xb_i)\}_{i = 1}^n$. In this section and further, if some variable has index~$n$, then this variable depends on training sample (\ref{eq:training_sample}) of size~$n$.

Define the {\it information matrix} $A_n \in \R^{P \times P}$ as
\begin{equation}
\label{eq:information_matrix}
A_n = \sum_{i = 1}^{n}{\boldsymbol \Psi (\xb_i) \boldsymbol \Psi ^ T (\xb_i)}.
\end{equation}
Then, we can obtain asymptotic properties of the indices estimates (\ref{eq:sobol_pce_approximation}) based on model (\ref{eq:PCE_and_noise}) while new data points $\{\xb_n,y_n = f(\xb_n)\}$ are added to the training sample sequentially. In order to prove these asymptotic properties we \textit{require only} that  $\varepsilon = \varepsilon(\vec{X})$ and $\{\Psi_j(\vec{X})\}_{j=0}^{P-1}$ are \textit{orthogonal} w.r.t. the distribution $\DD$, and we \textit{do not need to require} that multivariate polynomials$\{\Psi_{\al}(\vec{X}), \; \al \in \rr\}$ are \textit{orthonormal}.

\begin{theorem}
\label{th:asymptotic_theorem}
Let the following assumptions hold true:
\begin{enumerate}
\item We assume that there is an infinite sequence of points in the design space $\{\xb_i \in \XX\}_{i=1}^{\infty}$, generated by the corresponding sequence of i.i.d. random vectors, such that a.s.
\begin{equation}
\label{eq:limit_information_matrix}
\frac{1}{n} A_n = \frac{1}{n} \sum_{i = 1}^{n}{\boldsymbol \Psi (\xb_i) \boldsymbol \Psi ^ T (\xb_i)}  \ton \Sigma,
\end{equation}
where $\Sigma \in \R^{P \times P}$, where $\Sigma$ is a symmetric and non-degenerate matrix ($\Sigma = \Sigma^T$ and $det\Sigma>0$), and new design points are added successively from this sequence to the design of experiments $X_n = \{\xb_i\}_{i=1}^n$.
\item Let the vector-function be defined by its components according to (\ref{eq:sobol_pce}):
\[
\Sb(\nub) = (S_1(\nub), \ldots, S_d(\nub))^T
\]
and $\hat{\Sb}_n \triangleq \Sb(\hat{\cb}_n)$, where $\hat{\cb}_n$ is defined by (\ref{eq:coefficients_ls_opt_problem}).

\item Assume that for the true coefficients $\cb$ of model (\ref{eq:PCE_and_noise}):
\begin{equation}
\label{eq: detBSB}
\det (B \Sigma ^{- 2}\Gamma B ^ T) \neq 0,
\end{equation}
where $B$ is the matrix of partial derivatives defined as
\begin{equation}
\label{eq: derivB}
B \triangleq  B(\cb) = \left. \frac{\partial \Sb(\nub)}{\partial \nub}\right|_{\nub=\cb} \in \R^{d \times P},
\end{equation}
and $\Gamma = (\gamma_{r,s})_{r,s=0}^{P-1}\in\R^{P \times P}$ with $\gamma_{r,s} = \ee\left(\varepsilon^2\Psi_r(\vec{X})\Psi_s(\vec{X})\right)$,
\end{enumerate}
then
\begin{equation}
\label{eq:S_asymptotic}
 \sqrt{n} \left(\Sb(\hat{\cb}_n) - \Sb({\cb}) \right)  \tod \mathcal{N} (0, B \Sigma ^{- 2}\Gamma B ^ T).
\end{equation}
\end{theorem}
\begin{proof}
Let us denote by $\boldsymbol{\varepsilon}_n = (\varepsilon_1,\ldots,\varepsilon_n)^T\in\R^{n}$ the column vector, generated by the i.i.d. residual process values (see \eqref{eq:PCE_and_noise}), and by $\boldsymbol{\Psi}_n = (\boldsymbol{\Psi}(\xb_1),\ldots,\boldsymbol{\Psi}(\xb_n))\in\R^{P \times n}$ the design matrix. We can easily get that
\[
\hat{\cb}_n = A_n^{-1}\boldsymbol{\Psi}_nY_n = \cb + \left(\frac{1}{n}A_n\right)^{-1}\left[\frac{1}{n}\boldsymbol{\Psi}_n\boldsymbol{\varepsilon}_n\right].
\]

We can represent $\frac{1}{n}\boldsymbol{\Psi}_n\boldsymbol{\varepsilon}_n$ as $\frac{1}{n}\sum_{i=1}^n\boldsymbol{\xi}_i$, where $(\boldsymbol{\xi}_i)_{i=1}^n$ is a sequence, generated by i.i.d. random vectors $\boldsymbol{\xi}_i = \varepsilon(\vec{X}_i)\boldsymbol{\Psi}(\vec{X}_i)\in\R^P$, $i = 1,\ldots,n$, such that
$\ee\boldsymbol{\xi}_i = 0$ thanks to the fact that $\varepsilon$ and $\Psi_k(\vec{X})$ are orthogonal for $k<P$, and $\V[\boldsymbol{\xi}_i] = \Gamma$. 

Thus from (\ref{eq:limit_information_matrix}) and the central limit theorem we get that
\[
\sqrt{n}(\hat{\cb}_n - \cb) = \left(\frac{1}{n}A_n\right)^{-1}\left[\frac{1}{\sqrt{n}}\sum_{i=1}^n\boldsymbol{\xi}_i\right] \tod \Nd(0, \;  \Sigma^{-2}\Gamma).
\]

Applying $\delta$-method (see \cite{OehlertDelta1992}) to the vector-function $\Sb(\nub)$ at the point $\nub = \cb$, we obtain required asymptotics (\ref{eq:S_asymptotic}).
\end{proof}

\begin{remark}
Note that the elements of $B$ have the following form
\begin{equation}
\label{eq: derivB_terms}
b_{i\be} \triangleq \frac{\partial S_i}{\partial c_{\be}} = 
\begin{cases}

\frac{2c_{\be} \sum_{\al \in \rrst} c_{\al}^2  -2c_{\be} \sum_{\al \in \rr_i} c_{\al}^2}{\left(\sum_{\al \in \rrst} c_{\al}^2\right)^2}, & \text{if} \; \be \in \rr_i, \\

0, & \text{if} \; \be = \mathbf{0} \triangleq \{0, \ldots, 0\}, \\

\frac{-2c_{\be} \sum_{\al \in \rr_i} c_{\al}^2}{\left(\sum_{\al \in \rrst} c_{\al}^2\right)^2}, & \text{if} \; \be \notin \rr_i \cup \mathbf{0},
\end{cases}
\end{equation}
where $i = 1, \ldots, d$ and multi-index $\be \in \rr$. The elements of $B$ can be also represented as
\begin{equation}
\label{eq: derivB_terms_form2}
b_{i\be} \triangleq \frac{\partial S_i}{\partial c_{\be}} = \frac{-2c_{\be} }{\sum_{\al \in \rrst} c_{\al}^2} \times
\begin{cases}

S_i - 1, & \text{if} \; \be \in \rr_i, \\

0, & \text{if} \; \be = \mathbf{0} \triangleq \{0, \ldots, 0\}, \\

S_i, & \text{if} \; \be \notin \rr_i \cup \mathbf{0},
\end{cases}
\end{equation}


\end{remark}

\begin{remark}
We can see that conditions of the theorem do not depend on the type of orthonormal  polynomials.
\end{remark}

\begin{remark}
In case $\{\Psi_{\al}(\vec{X}), \; \al \in \rr\}$ are multivariate polynomials, orthonormal w.r.t. the distribution $\DD$, we get that $\Sigma = I\in\R^{P \times P}$ is the identity matrix.
\end{remark}

\begin{remark}
In the proof of theorem \ref{th:asymptotic_theorem} we are trying to make as less assumptions as possible in order to depart from original polynomial chaos model \eqref{eq:pce_full} as little as possible. That is why the only important assumption is that  $\varepsilon = \varepsilon(\vec{X})$ and $\{\Psi_j(\vec{X})\}_{j=0}^{P-1}$ are \textit{orthogonal} w.r.t. the distribution $\DD$. However, we can also consider model \eqref{eq:PCE_and_noise} as a regression one, and so the error term $\varepsilon$ is modelled by a white noise, independent from $\{\Psi_j(\vec{X})\}_{j=0}^{P-1}$, see the discussion of the polynomial chaos approach from a statistician's perspective in \cite{ohagan13-polyn-chaos}. Nevertheless, even in the case of such interpretation of model \eqref{eq:PCE_and_noise} we still get the same asymptotic behavior \eqref{eq:S_asymptotic}.
\end{remark}

\begin{remark}
\label{importantLab}
In case $\varepsilon$ and $\Psi_k(\vec{X})$ are not only orthogonal for $k<P$, but also are independent, we get that $\Gamma=\sigma^2\Sigma$. Then asymptotics \eqref{eq:S_asymptotic} takes the form
\begin{equation}
\label{asymptoticLab}
 \sqrt{n} \left(\Sb(\hat{\cb}_n) - \Sb({\cb}) \right)  \tod \mathcal{N} (0, \sigma ^ 2 B \Sigma ^{- 1} B ^ T).
\end{equation}
\end{remark}

In applications it seems reasonable to assume that $\varepsilon$ and $\Psi_k(\vec{X})$ are approximately independent for $k<P$. Then for practical purposes we can use asymptotics \eqref{asymptoticLab}, for which it is easier to calculate the asymptotic covariance matrix. Therefore in the sequel \textit{for applications we are going to use this simplified expression}.

\section{Design of Experiments Construction}
\label{sec: PCEOptimalDoE}

\subsection{Preliminary Considerations}

Taking into account the results of Theorem~\ref{th:asymptotic_theorem}, the limiting covariance matrix of the indices estimates depends on
\begin{enumerate}
\item Noise variance $ \sigma ^ 2 $,
\item True values of PC coefficients $\cb$, defining $B$,
\item Experimental design $X$, defining $\Sigma$.
\end{enumerate}

If we have a sufficiently accurate approximation of the original model, then in the above assumptions, asymptotic covariance in \eqref{asymptoticLab} provides a \textit{theoretically motivated functional to characterize the quality of the experimental design}. Indeed, generally speaking the smaller the norm of the covariance matrix $\|\sigma ^ 2 B \Sigma ^{- 1} B ^ T\|$, the better the estimation of the sensitivity indices apparently should be. Theoretically, we could use this formula for constructing an experimental design that is effective for calculating Sobol' indices: we could select designs that minimize the norm of the covariance matrix. However, there are some problems when proceeding this way:
\begin{itemize}
\item The first one relates to selecting some specific functional for minimization. Informally speaking, we need to choose ``the norm'' associated with the limiting covariance matrix;

\item The second one refers to the fact that we do not know true values of the PC model coefficients, defining $B$; therefore, we will not be able to accurately evaluate the quality of the design.

\end{itemize}

The first problem can be solved in different ways. A~number of statistical criteria for design optimality ($D$-, $I$-optimality and others, see~\cite{Chaloner95bayesianexperimental}) are known. Similar to the work \cite{BurnaevPanin15}, we use the $D$-optimality criterion, as it a provides computationally efficient procedure for design construction. $D$-optimal experimental design minimizes the determinant of the limiting covariance matrix. If the vector of the estimated parameters is normally distributed then $D$-optimal design allows to minimize the volume of the confidence region for this vector.

The second problem is more complex. The optimal design for estimating sensitivity indices that minimizes the norm of limiting covariance matrix depends on true values of the indices, so it can be constructed only if these true values are known. However, in this case design construction makes no sense.

The dependency of the optimal design for indices evaluation on the true model parameters is a consequence of the indices estimates nonlinearity w.r.t. the PC model coefficients. In order to underline this dependency, the term ``{\it locally} $D$-optimal design'' is commonly used \cite{PronzatoAdDoptimal2010}.  In this setting there are several approaches, which are usually associated with either some assumptions about the unknown parameters, or adaptive design construction (see \cite{PronzatoAdDoptimal2010}). We use the latter approach.


In the case of adaptive designs, new design points are generated sequentially based on current estimates of the unknown parameters. This allows to avoid prior assumptions on these parameters. However, this approach has a problem with a confidence of the solution found: if at some step of the design construction process parameters estimates are significantly different from their true values, then the design, which is constructed based on these estimates, may lead to new parameters estimates, which are even more different from the true values. 


In practice, during the construction of adaptive design, the quality of the approximation model and assumptions on non-degeneracy of results can be checked at each iteration and one can control and adjust the adaptive strategy.

\subsection{Adaptive DoE Algorithm}

In this section, we introduce the adaptive algorithm for constructing a design of experiments that is effective to estimate sensitivity indices based on the asymptotic $D$-optimality criterion (see description of Algorithm \ref{alg: designBuilding} and its scheme in Figure~\ref{pic:AlgScheme}). As it was discussed, the main idea of the algorithm is to minimize the confidence region for indices estimates. At each iteration, we replace the limiting covariance matrix by its approximation based on the current PC coefficients estimates.

\begin{figure}[h!]
\centering
\includegraphics[scale = 0.20]{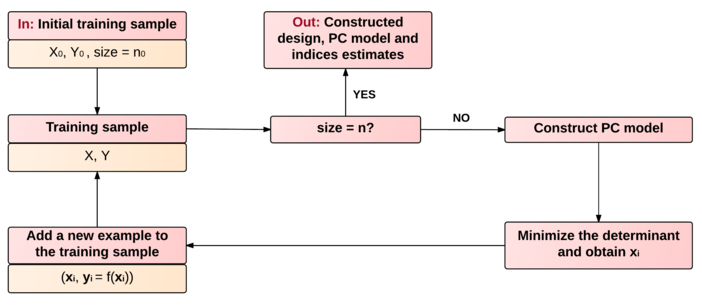}
\caption{\label{pic:AlgScheme} Adaptive algorithm for constructing an effective experimental design to evaluate PC-based Sobol' indices}
\end{figure}

As for initialization, we assume that there is some {\it initial design}, and we require that this initial design is non-degenerate, {\it i.e.} such that the initial information matrix $A_0$ is nonsingular ($\det A_0 \neq 0$). In addition, at each iteration the non-degeneracy of the matrix $B_i A_i^{-1} B_i^T$, related to the criterion to be minimized, is checked.

\begin{algorithm}
\begin{description}
\item [{Goal:}] Construct an effective experimental design for the calculation of sensitivity indices
\item [{Parameters:}] initial and final numbers of points $n_0$ and $n$ in the design; set of candidate design points $\Xi$.
\item [{Initialization:}] \end{description}
\begin{itemize}
\item  initial training sample $\{X_0, Y_0\}$ of size $n_0$, where design $ X_0 = \{\xb_i \}_{i = 1} ^{n_0} \subset \Xi $ defines a non-degenerate information matrix $ A_0 = \sum_{i = 1} ^{n_0}{\boldsymbol \Psi (\xb_i) \boldsymbol \Psi ^ T (\xb_i)} $;
\item $ B_0 = B(\hat{\cb}_0) $, obtained using the initial estimates of the PC model coefficients, see (\ref{eq: derivB}), (\ref{eq: derivB_terms}), (\ref{eq: derivB_terms_form2});
\end{itemize}
\begin{description}
\item [{Iterations:}] for all $ i $ from 1 to $n - n_0 $: \end{description}
\begin{itemize}
\item Solve optimization problem
\begin{equation}
\label{opt_prob}
 \xb_{i} = {\arg\min}_{\xb \in \Xi} \: \det \left [B_{i-1} (A_{i-1} + \boldsymbol \Psi (\xb) \boldsymbol \Psi ^ T (\xb)) ^{- 1} B_{i-1} ^ T \right] 
\end{equation}

\item $ A_{i} = A_{i-1} + \boldsymbol \Psi (\xb_{i}) \boldsymbol \Psi ^{T} (\xb_{i}) $

\item Add the new sample point $(\xb_{i}, y_i = f(\xb_{i}))$ to the training sample and update current estimates $\hat{\cb}_i$ of the PCE model coefficients

\item Calculate $ B_i = B(\hat{\cb}_i) $

\end{itemize}
\begin{description}
\item [{Output:}] The design of experiments $ X = X_0 \cup X_{add} $, where $ X_{add} = \{\xb_k \}_{k = 1}^{n-n_0} $, $Y = f(X)$ \end{description}
\caption{Description of the Adaptive DoE algorithm}
\label{alg: designBuilding}
\end{algorithm}

\subsection{Details of the Optimization procedure}
\label{subsec: PCEOptimalDoEformula}

The idea behind the proposed optimization procedure is analogous to the idea of the Fedorov's algorithm for constructing optimal designs \cite{FedorovAlgDOptimal1994}. In order to simplify optimization problem \eqref{opt_prob}, we use two well-known identities:

\begin{itemize}
\item Let $M$ be some nonsingular square matrix, $ \tb $ and $ \wb $ be vectors such that $ 1 + \wb ^ T M^{-1} \tb \neq 0 $, then
\begin{equation}
\label{eq: PCEInv}
(M + \tb \wb ^ T) ^{- 1} = M ^{- 1} - \frac{M ^{- 1} \tb \wb ^ TM ^{- 1}}{1 + \wb ^ TM ^{- 1} \tb}.
\end{equation}

\item Let $M$ be some nonsingular square matrix, $ \tb $ and $ \wb $ be vectors of appropriate dimensions, then
\begin{equation}
\label{eq: PCEDet}
\det (M + \tb \wb ^ T) = \det (M) \cdot (1 + \wb ^ TM ^{- 1} \tb).
\end{equation}

\end{itemize}

Let us define $ D \triangleq B (A + \phx \phTx) ^{- 1} B ^ T $, then applying (\ref{eq: PCEInv}), we obtain
\[
\det (D) = \det \left [B A^{- 1} B ^ T - \frac{B A^{- 1} \phx \phTx A^{- 1} B ^ T}{1 + \phTx A ^{- 1} \phx} \right] \triangleq 
\]
\[
\triangleq \det \left [M - \tb \wb ^ T \right],
\]
where $ M \triangleq B A ^{- 1} B ^ T $, $ \tb \triangleq \frac{B A ^{- 1} \phx}{1 + \phTx A ^{- 1} \phx} $, $ \wb \triangleq B A ^{- 1} \phx $. Assuming that matrix $M$ is nonsingular and applying (\ref{eq: PCEDet}), we obtain
\[
\det (D) = \det (M) \cdot (1 - \wb ^ TM ^{- 1} \tb) \to \min.
\]
The resulting optimization problem is
\[
\wb ^ T M ^{- 1} \tb \to \max, \; \; \; \;
\]
or explicitly \eqref{opt_prob} is reduced to
\[
  \frac{(\phTx A ^{- 1}) B ^ T (B A ^{- 1} B ^ T) ^{- 1} B (A ^{- 1} \phx)}{1 + \phTx A ^{- 1} \phx} \to \max_{\xb \in \Xi}.
\]

\section{Benchmark}
\label{sec: SAPCEresults}

In this section, we validate the proposed algorithm on a set of computational models with different input dimensions. Several  analytic problems and two industrial problems based on finite element models are considered. Input parameters (variables) of the considered models have independent uniform and independent normal distributions. For some models, additionally independent gaussian noise is added to their outputs.

At first, we form non-degenerate random initial design, and then we use various techniques to add new design points iteratively. We compare our method for design construction (denoted as {\bf Adaptive for SI}) with the following methods:

\begin{itemize}

\item {\bf Random} method iteratively adds new design points randomly from the set of candidate design points $\Xi$;

\item {\bf Adaptive D-opt} iteratively adds new design points that maximize the determinant of information matrix (\ref{eq:information_matrix}): $\det A_n \to \max_{\xb_n \in \Xi} $ (\cite{FedorovAlgDOptimal1994}). The resulting design is optimal, in some sense,  for estimation of the PCE model coefficients. We compare our method with this approach to prove that it gives some advantage over usual $D$-optimality. Strictly speaking, $D$-optimal design is not iterative but if we have an initial training sample then the sequential approach seems a natural generalization of a common $D$-optimal designs.

\item {\bf LHS}. Unlike other considered  designs, this method is not iterative as a completely new design is generated at each step. This method uses Latin Hypercube Sampling, and it is common to compute PCE coefficients.
\end{itemize}

The metric of design quality is {\bf the mean error} defined as the distance between estimated and true indices $\sqrt{\sum_{i=1}^d (S_i - \hat{S}_i^{\text{\:run}})^2}$ averaged over runs with different random initial designs ($200$-$400$ runs). We consider not only the {\it mean} error but also its {\it variance}. Particularly, we use Welch's t-test (see \cite{Welch1947}) to ensure that the difference of mean distances is statistically significant for the considered methods. Note that lower p-values correspond to bigger confidence.

In all cases, we assume that the truncation set (retained PCE terms) is selected before an experiment.

\subsection{Analytic Functions}

\paragraph{The Sobol' function} is commonly used for benchmarking methods in global sensitivity analysis

\[
f(\xb) = \prod_{i=1}^d \frac{|4x_i-2|+c_i}{1+c_i},
\]
where $x_i \sim \dU(0,1)$. In our case parameters $d=3$, $c = (0.0, 1.0, 1.5)$ are used. Independent gaussian noise is added to the output of the function. The standard deviation of noise is $0$ (without noise), $0.2$ and $1.4$ that corresponds to $0\%$, $28\%$ and $194\%$ of the function standard deviation, caused by the inputs uncertainty. Analytical expressions for the corresponding sensitivity indices are available in \cite{Sobol93}.

\paragraph{Ishigami function} is also commonly used for benchmarking of global sensitivity analysis: 

\[
f(\xb) = \sin x_1 + a \sin^2 x_2 + b x_3^4 \sin x_1, \; a=7, \; b=0.1
\]
where $x_i \sim \dU(-\pi, \pi)$. Theoretical values for its sensitivity indices are available in \cite{Marrel2009}.

\paragraph{Environmental function} models a pollutant spill caused by a chemical accident \cite{Bliznyuk08}

\[
f(\xb) = \sqrt{4\pi} C(\xb), \; \xb = (M, \;  d,\; L,\; \tau), 
\]

\[
C(\xb) = \frac{M}{ \sqrt{4\pi D t}} \exp \left(\frac{-s^2}{4 D t} \right) + \frac{M}{\sqrt{4\pi D (t-\tau)}} \exp \left(  
- \frac{(s-L)^2}{4D(t-\tau)}\right) I(\tau<t),
\]
where $I$ is the indicator function; $4$ input variables and their distributions are defined as: 
$M$ $\sim$ $\dU(7, 13)$,	mass of pollutant spilled at each location;
$D$ $\sim$ $\dU(0.02, 0.12)$,	diffusion rate in the channel;
$L$ $\sim$ $\dU(0.01, 3)$,	location of the second spill;
$\tau$ $\sim$ $\dU(30.01, 30.295)$,   	time of the second spill. $C(\xb)$ is the concentration of the pollutant at the space-time vector $(s, \; t)$, where $0 \leq s \leq 3 $ and $t > 0$.  

We consider a cross-section corresponding to $t=40$, $s=1.5$ and suppose that independent gaussian noise $\mathcal{N}(0, 0.5^2)$ is added to the output of the function.

\paragraph{The Borehole function} models water flow through a borehole. It is commonly used for testing different methods in numerical experiments \cite{Worley87,Moon12}
\[
f(\xb) = \frac {2\pi T_u (H_u - H_l)} {\ln(r/r_w)(1+ \frac{2LT_u}{\ln(r/r_w)r_w^2K_w} + T_u/T_l)},
\]
where $8$ input variables and their distributions are defined as: 
$r_w$ $\sim$ $\dU(0.05, 0.15)$, radius of borehole (m);
$T_u$ $\sim$ $\dU(63070, 115600)$, transmissivity of upper aquifer ($m^2$/yr); 
$r$ $\sim$ $\dU(100, 50000)$, radius of influence (m);
$H_u$ $\sim$ $\dU(990, 1110)$,	potentiometric head of upper aquifer (m); 
$T_l$ $\sim$ $\dU(63.1, 116)$,	transmissivity of lower aquifer ($m^2$/yr);
$H_l$ $\sim$ $\dU(700, 820)$,	potentiometric head of lower aquifer (m);
$L$ $\sim$ $\dU(1120, 1680)$,	length of borehole (m);
$K_w$ $\sim$ $\dU(9855, 12045)$, hydraulic conductivity of borehole (m/yr).

Besides the deterministic case, we also consider stochastic one when independent gaussian noise $\mathcal{N}(0, 5.0^2)$ is added to the output of the function.

\paragraph{The Wing Weight function} models weight of an aircraft wing \cite{Forrester08}
\[
f(\xb) = 0.036 S_w^{0.758} W_{fw}^{0.0035}\left(\frac{A}{\cos^2(\Lambda)}\right)^{0.6} q^{0.006} \lambda^{0.04} \left(\frac{100 t_c}{\cos(\Lambda)}\right)^{-0.3} (N_z W_{dg})^{0.49} +
\]
\[
+S_w W_p,
\]
where  $10$ input variables and their distributions are defined as: 
$S_w$ $\sim$ $\dU(150, 200)$, wing area ($ft^2$);
$W_{fw}$ $\sim$ $\dU(220, 300)$, weight of fuel in the wing (lb);
$A$ $\sim$ $\dU(6, 10)$, aspect ratio;
$\Lambda$ $\sim$ $\dU(-10, 10)$, quarter-chord sweep (degrees);
$q$ $\sim$ $\dU(16, 45)$, dynamic pressure at cruise (lb/$ft^2$);
$\lambda$ $\sim$ $\dU(0.5, 1)$, taper ratio;
$t_c$ $\sim$ $\dU(0.08, 0.18)$, aerofoil thickness to chord ratio;
$N_z$ $\sim$ $\dU(2.5, 6)$, ultimate load factor;
$W_{dg}$ $\sim$ $\dU(1700, 2500)$, flight design gross weight (lb);
$W_p$ $\sim$ $\dU(0.025, 0.08)$, paint weight (lb/$ft^2$). 

Besides the deterministic case, we also consider stochastic one when independent gaussian noise $\mathcal{N}(0, 5.0^2)$ is added to the output of the function.

\paragraph{Experimental setup:} In the experiments, we assume that the set of candidate design points $\Xi$ is a uniform grid in the $d$-dimensional hypercube. Note that $\Xi$ affects optimization quality. Experimental settings for analytical functions are summarized in Table~\ref{tab:analyt_func_exp_cond}.

\begin{table}[hbt]
\caption{Benchmark settings for analytical functions}
\label{tab:analyt_func_exp_cond}
\begin{center}
\begin{tabular}{lccccc} 
\hline
Characteristic & Sobol & Ishigami & Environmental & Borehole & WingWeight\\
\hline
Input dimension    & $3$ & $3$ & $4$ & $8$ & $10$\\
Input distributions    & Unif & Unif & Unif   & Unif  & Unif\\
PCE degree & $9$  &   $9$ & $5$ & $4$   &   $4$ \\
$q$-norm      & $0.75$ & $0.75$  & $1$ & $0.75$ & $0.75$\\
Regressors number    & $111$ & $111$   & $126$  & $117$ &   $176$\\
Initial design size   & $150$  & $120$ & $126$ & $117$ & $186$\\
Added noise std   & (0, 0.2, 1.4) & ~--- & $0.5$ & ($0$, $5.0$) & ($0$, $5.0$)\\
\hline
\end{tabular}
\end{center}
\end{table}

\subsection{Finite Element Models}

\paragraph{Case 1: Truss model.} The deterministic computational model, originating from \cite{Lee06}, resembles the displacement $V_1$ of a truss structure with $23$ members as shown in Figure~\ref{pic:Truss}. 

\begin{figure}
\centering
\includegraphics[scale = 0.30]{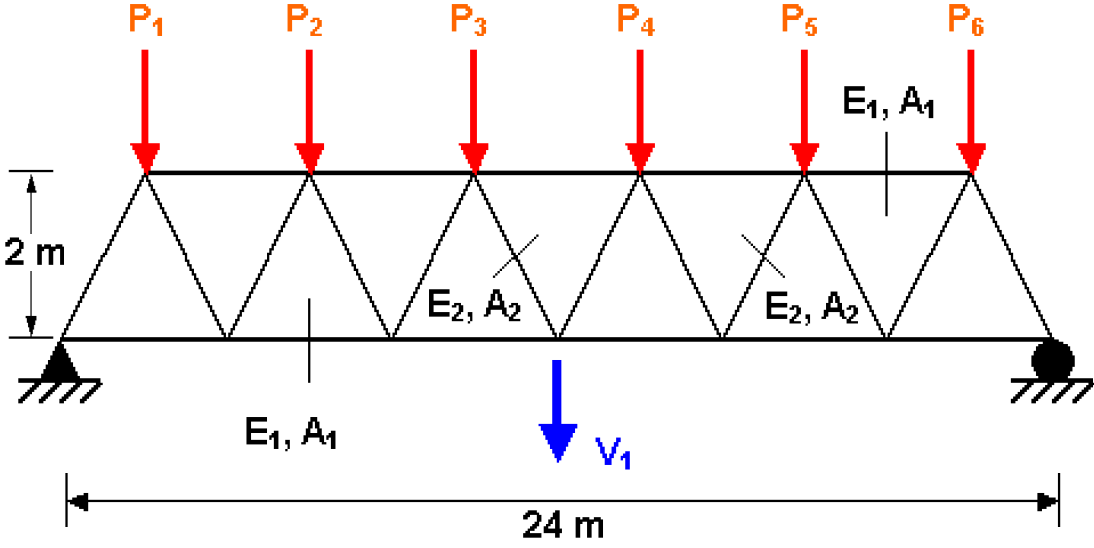}
\caption{\label{pic:Truss} Truss structure with $23$ members}
\end{figure}

Ten random variables are considered: 
\begin{itemize}
\item $E_1$, $E_2$ (Pa) $\sim$ $\dU(1.68 \times 10^{11}, \; 2.52 \times 10^{11})$;
\item $A_1$ ($m^2$) $\sim$ $\dU(1.6 \times 10^{-3}, \; 2.4 \times 10^{-3})$; 
\item $A_2$ ($m^2$) $\sim$ $\dU(0.8 \times 10^{-3}, \; 1.2 \times 10^{-3})$; 
\item $P_1$ - $P_6$ (N) $\sim$ $\dU(3.5 \times 10^{4}, \; 6.5 \times 10^{4})$.
\end{itemize}
   
It is assumed that all the horizontal elements have perfectly correlated Young's modulus and cros-sectional areas with each other and so is the case with the diagonal members.

\paragraph{Case 2: Heat transfer model.} We consider the two-dimensional stationary heat diffusion problem described in~\cite{Konakli15}.
The problem is defined on the square domain $D = (-0.5, 0.5) \times (-0.5, 0.5)$ shown in Figure~\ref{pic:Heat}, where the temperature field $T(z), z \in D$ is described by the partial differential equation:
\[
-\nabla(\kappa(\mathbf{z})\nabla T(z)) = 500 I_A(z),
\]
with boundary conditions $T = 0$ on the top boundary and $\nabla T \mathbf{n} = 0$ on the left, right and bottom boundaries, where $\mathbf{n}$ denotes the vector normal to the boundary; $A = (0.2, 0.3)\times(0.2, 0.3)$ is a square domain within $D$ and $I_A$ is the indicator function of $A$. The diffusion coefficient, $\kappa(z)$, is a lognormal random field defined by
\[
\kappa(z) = \exp[a_k + b_k g(z)],
\]
where $g(z)$ is a standard Gaussian random field and the parameters $a_k$ and $b_k$ are such that the mean and standard deviation of $\kappa$ are $\mu_{\kappa}= 1$ and $\sigma_\kappa= 0.3$, respectively. The random field $g(z)$ is characterized by an autocorrelation function $\rho(z, z^\prime) = \exp(-\|z-z^\prime\|^2/0.2^2)$. The quantity of interest, $Y$, is the average temperature in the square domain $B = (-0.3,-0.2)\times (-0.3,-0.2)$ within
$D$ (see Figure~\ref{pic:Heat}).

To facilitate solution of the problem, the random field $g(z)$ is represented using the Expansion Optimal Linear Estimation (EOLE) method (see \cite{Li93}). By truncating the EOLE series after the first $M$ terms, $g(z)$ is approximated by 
\[
\hat{g}(z) = \sum_{i=1}^{M} \frac{\xi_i}{\sqrt{\ell_i}} \phi_i^T \mathbf{C}_{z\zeta}.
\]
In the above equation, $\{\xi_1, \ldots, \xi_M\}$ are independent standard normal variables; $\mathbf{C}_{z\zeta}$ is a vector with elements $\mathbf{C}_{z\zeta}^{(k)} = \rho(z, \zeta_k)$, where $\{\zeta_1, \ldots, \zeta_M\}$ are the points of an appropriately defined mesh
in $D$; and $(\ell_i, \phi_i)$ are the eigenvalues and eigenvectors of the correlation matrix $\mathbf{C}_{\zeta\zeta}$ with elements $\mathbf{C}_{\zeta\zeta}^{(k,\ell)} = \rho(\zeta_k, \zeta_{\ell})$, where $k, \ell = 1, \ldots, n$. We select $M = 53$ in order to satisfy inequality
\[
\sum_{i=1}^M \ell_i / \sum_{i=1}^n \ell_i \geq 0.99.
\]

The underlying deterministic problem is solved with an in-house finite-element analysis code. The employed finite-element discretization with triangular $T3$ elements is
shown in Figure~\ref{pic:Heat}. Figure~\ref{pic:Heat1} shows the temperature fields corresponding to two example realizations of the diffusion coefficient.

%

\begin{figure}
    \centering
    \begin{subfigure}[b]{0.3\textwidth}
    \centering
        \includegraphics[scale = 0.2]{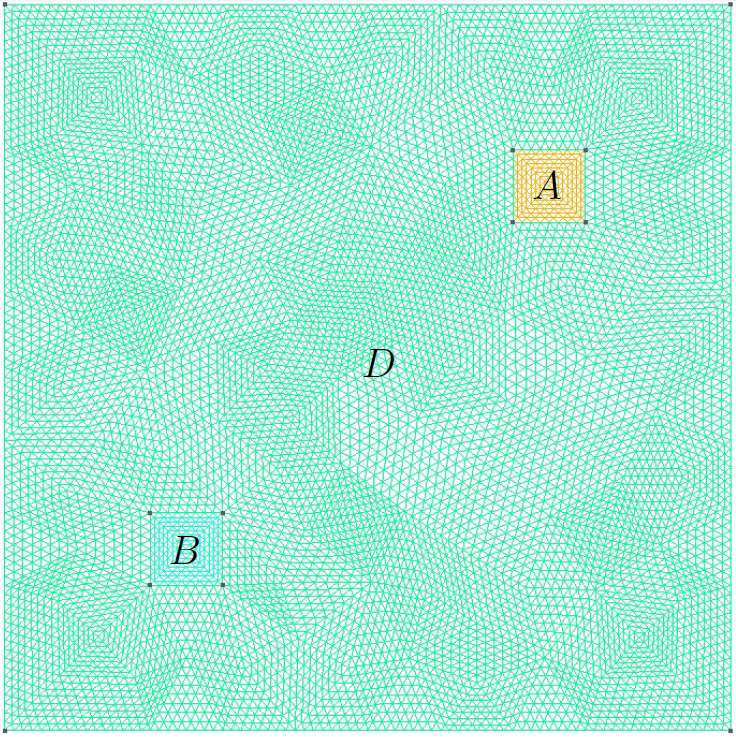}
        \caption{Finite-element mesh}
        \label{pic:Heat}
    \end{subfigure}
    ~ 
    \begin{subfigure}[b]{0.6\textwidth}
        \includegraphics[width=\textwidth]{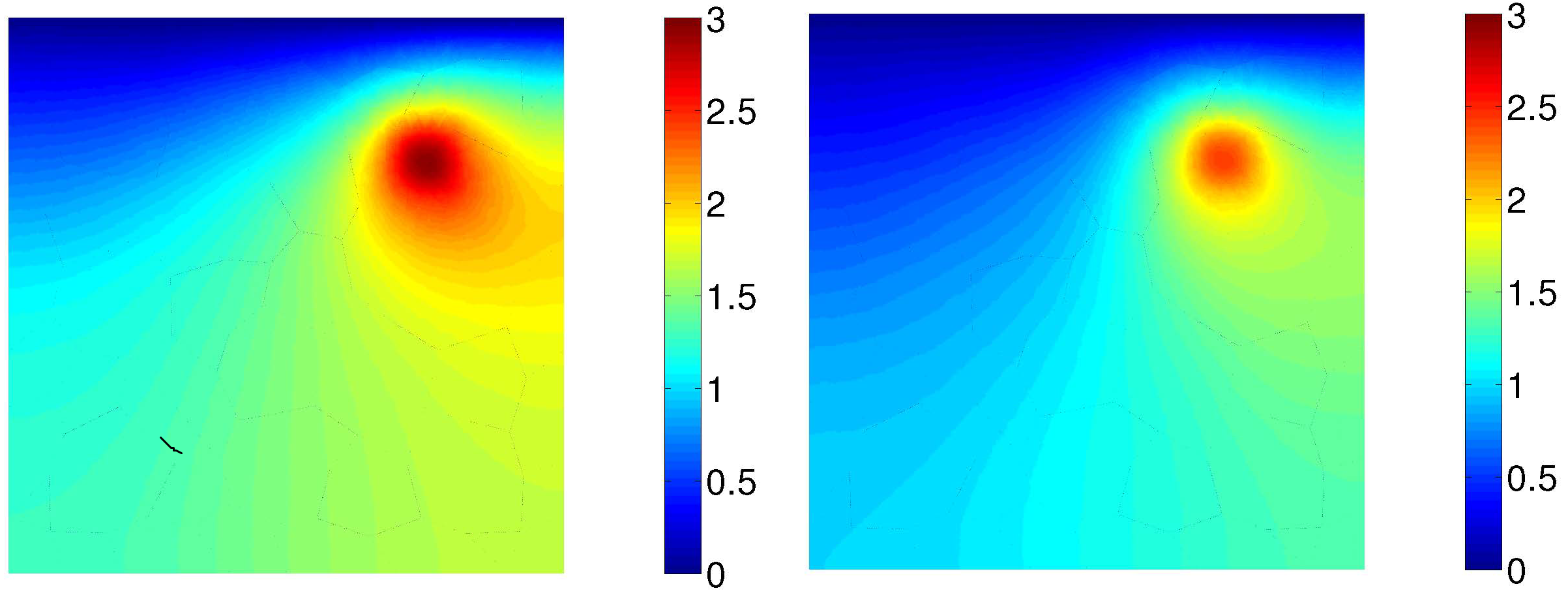}
        \caption{Temperature field realizations}
        \label{pic:Heat1}
    \end{subfigure}
    ~ 
    \caption{Heat diffusion problem}
\end{figure}

\paragraph{Experimental Setup.} For these finite element models, we assume that the set of candidate design points $\Xi$ is 
\begin{itemize}
\item a uniform grid in the $10$-dimensional hypercube for the Truss model;
\item LHS design with normally distributed variables in $53$-dimensional space for the Heat transfer model.
\end{itemize}

Experimental settings for all models are summarized in Table~\ref{tab:func_exp_cond}.

%
%


\begin{table}[hbt]
\caption{Benchmark settings for Finite Element models}
\label{tab:func_exp_cond}
\begin{center}
\begin{tabular}{lcc} 
\hline
Characteristic & Truss & Heat transfer\\
\hline
Input dimension     & $10$ & $53$\\
Input distributions   & Unif & Norm \\
PCE degree  & $4$ & $2$ \\
$q$-norm     & $0.75$ & $0.75$\\
Regressors number   & $176$ &   $107$   \\
Initial design size   & $176$ & $108$\\
Added noise std   & ~--- & ~--- \\
\hline
\end{tabular}
\end{center}
\end{table}

\subsection{Results}

Figures~\ref{pic:Sobol3}, \ref{pic:Sobol3_noise0-2}, \ref{pic:Sobol3_noise1-4},  \ref{pic:Ishigami3},  \ref{pic:Environmental4},  \ref{pic:Borehole8},  \ref{pic:Borehole8_noise5}, \ref{pic:WingWeight10}, \ref{pic:WingWeight10_noise5}   show results for analytic functions. Figures \ref{pic:Truss10Uniform} and \ref{pic:HeatProblem53} present results for finite element models. We provide here mean errors, relative mean errors {\it w.r.t} the proposed method and $p$-values to ensure that the difference of mean errors is statistically significant.

In the presented experiments, the proposed method performs better than other considered methods in terms of the mean error of estimated indices. Particularly note its superiority over standard LHS approach that is commonly used in practice. The difference in mean errors is statistically significant according to Welch's t-test.

Comparison of figures \ref{pic:Sobol3}, \ref{pic:Sobol3_noise0-2}, \ref{pic:Sobol3_noise1-4} with different levels of additive noise shows that the proposed method is effective when the analyzed function is deterministic or when the noise level is {\it small}.

Because of robust problem statement and limited accuracy of the optimization, the algorithm may produce duplicate design points. Actually, it's a common situation for locally $D$-optimal designs \cite{PronzatoAdDoptimal2010}. If the computational model is deterministic, one may modify the algorithm, {\it e.g.} exclude repeated design points.

Although high dimensional optimization problems may be computationally prohibitive, the proposed approach is still useful in high dimensional settings. We propose to generate a uniform candidate set ({\it e.g.} LHS design of large size) and then choose its subset for the effective calculation of Sobol' indices using our adaptive method, see results for Heat transfer model in Figure \ref{pic:HeatProblem53} (note that due to computational complexity we provide for this model results only for $2$ iterations of the LHS method).

It should be noted that in all presented cases the specification of sufficiently accurate PCE model (reasonable values for degree $p$ and $q$-norm defining the truncation set) is assumed to be known {\it a priori} and the size of the initial training sample is sufficiently large. If we use an inadequate specification of the PCE model ({\it e.g.} quadratic PCE in case of cubic analyzed function), the method will perform worse in comparison with methods which do not depend on PCE model structure. In any case, usage of inadequate PCE models may lead to inaccurate results. That is why it is very important to control PCE model error during the design construction. For example, one may use {\it cross-validation} for this purpose \cite{ElstatHastie2009}. Thus, if the PCE model error increases during design construction this may indicate that the model specification is inadequate and should be changed. 

\begin{figure}
    \centering
    \begin{subfigure}[b]{\textwidth}
        \includegraphics[width=\textwidth]{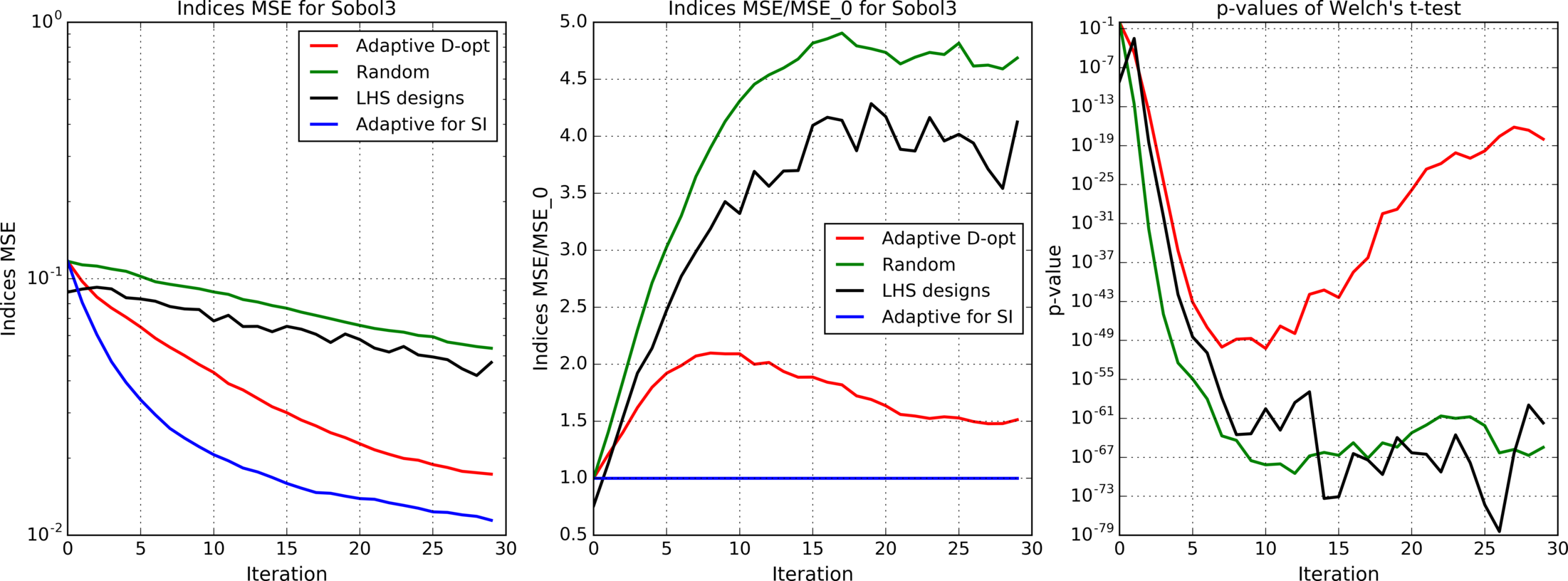}
        \caption{Noise std: 0}
        \label{pic:Sobol3}
    \end{subfigure}
    ~ 
    \begin{subfigure}[b]{\textwidth}
        \includegraphics[width=\textwidth]{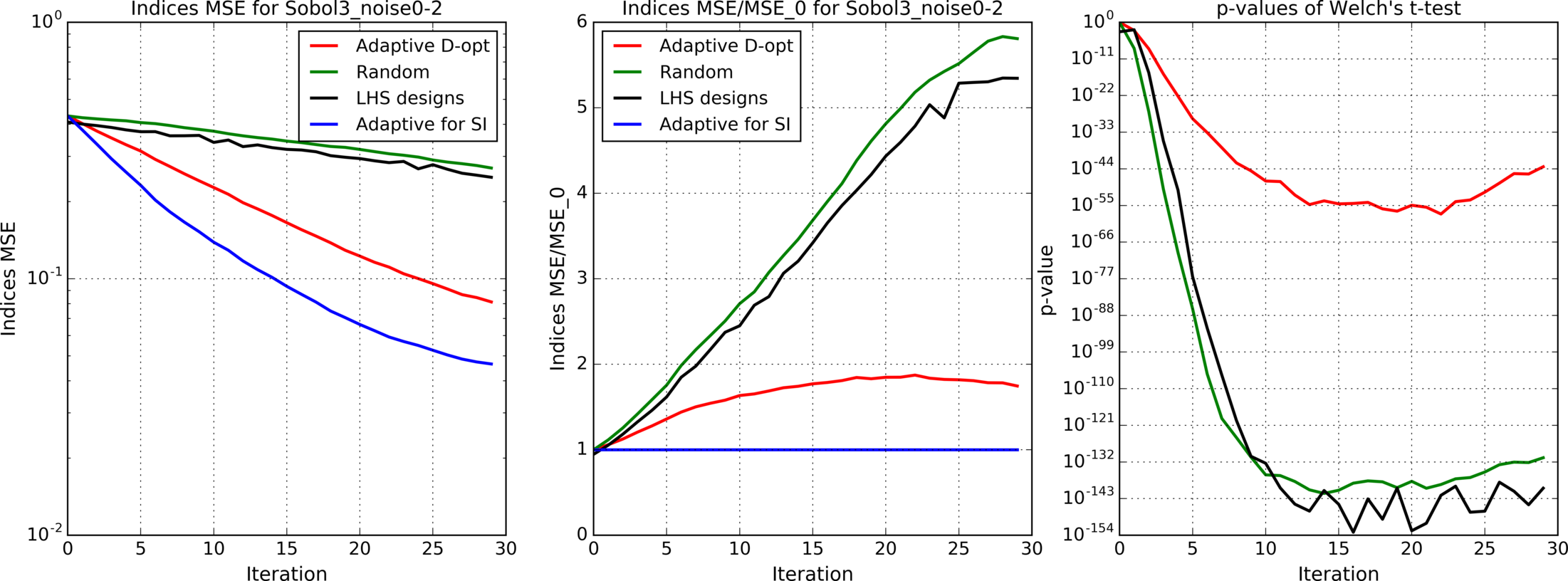}
        \caption{Noise std: 0.2}
        \label{pic:Sobol3_noise0-2}
    \end{subfigure}
    ~ 
    \begin{subfigure}[b]{\textwidth}
        \includegraphics[width=\textwidth]{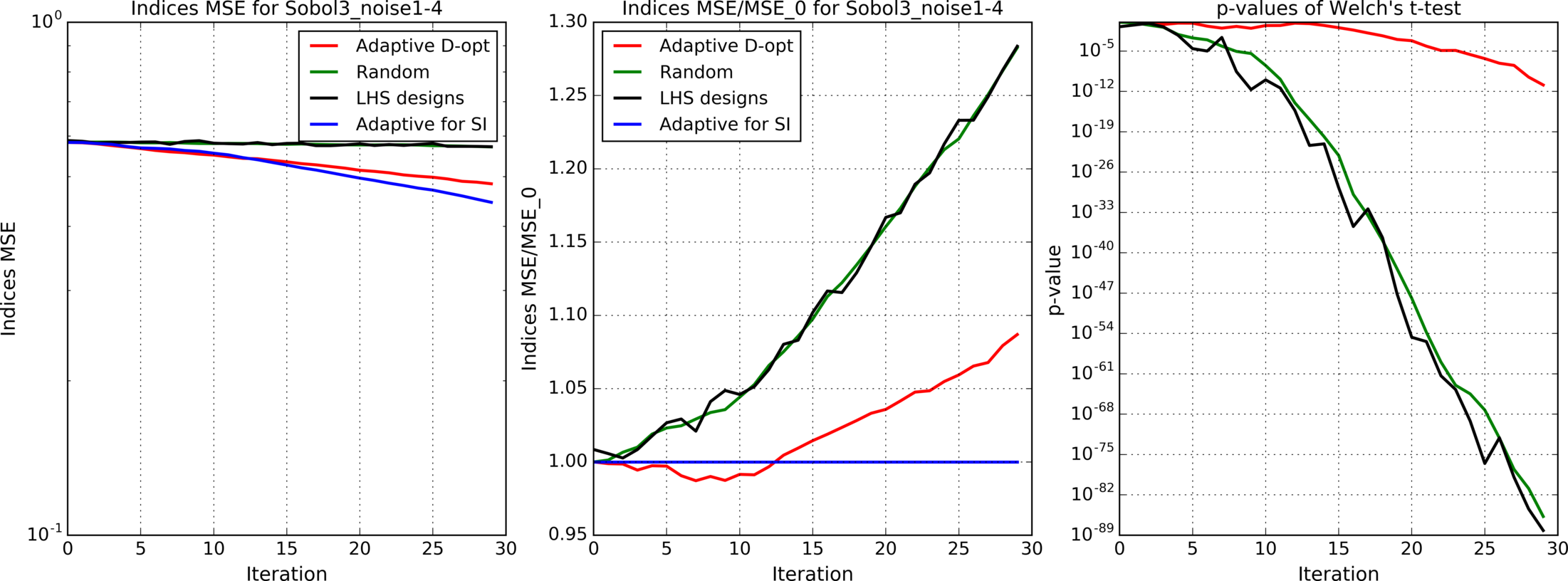}
        \caption{Noise std: 1.4}
        \label{pic:Sobol3_noise1-4}
    \end{subfigure}
    \caption{Sobol function. $3$-dimensional input.}
\end{figure}

\begin{figure}
\centering
\includegraphics[width=\textwidth]{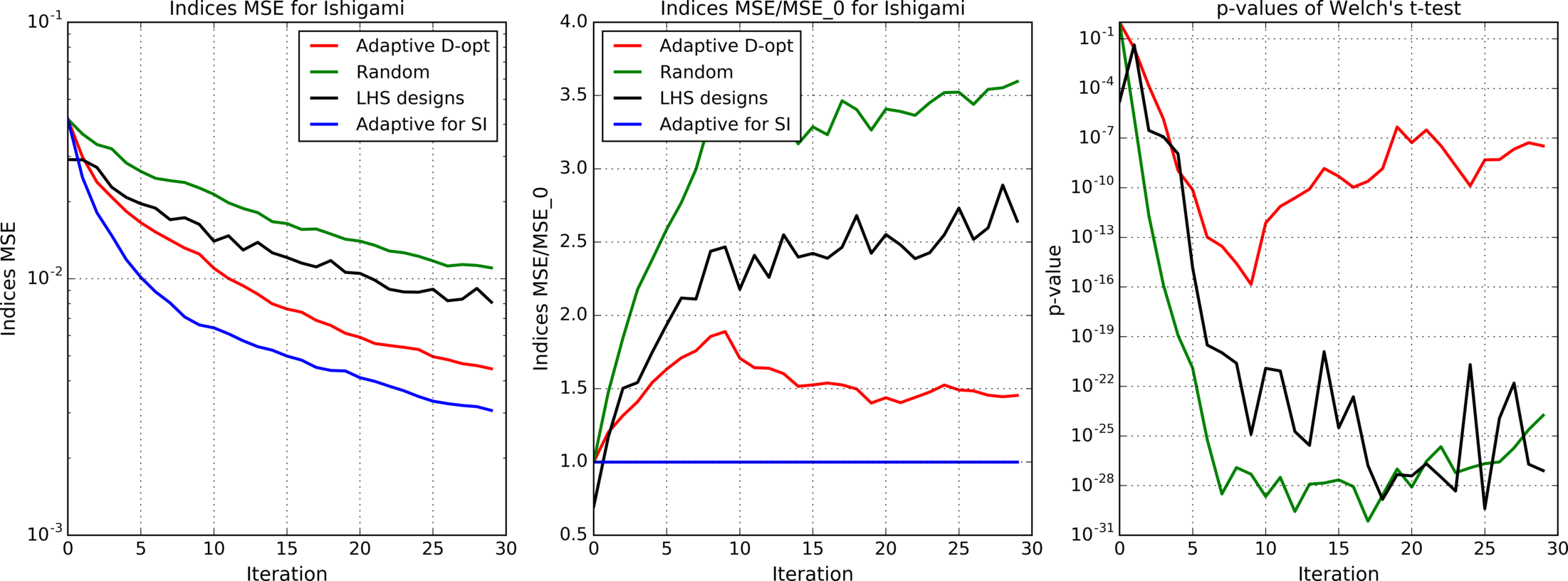}
\caption{\label{pic:Ishigami3} Ishigami function. $3$-dimensional input.}
\end{figure}

\begin{figure}
\centering
\includegraphics[width=\textwidth]{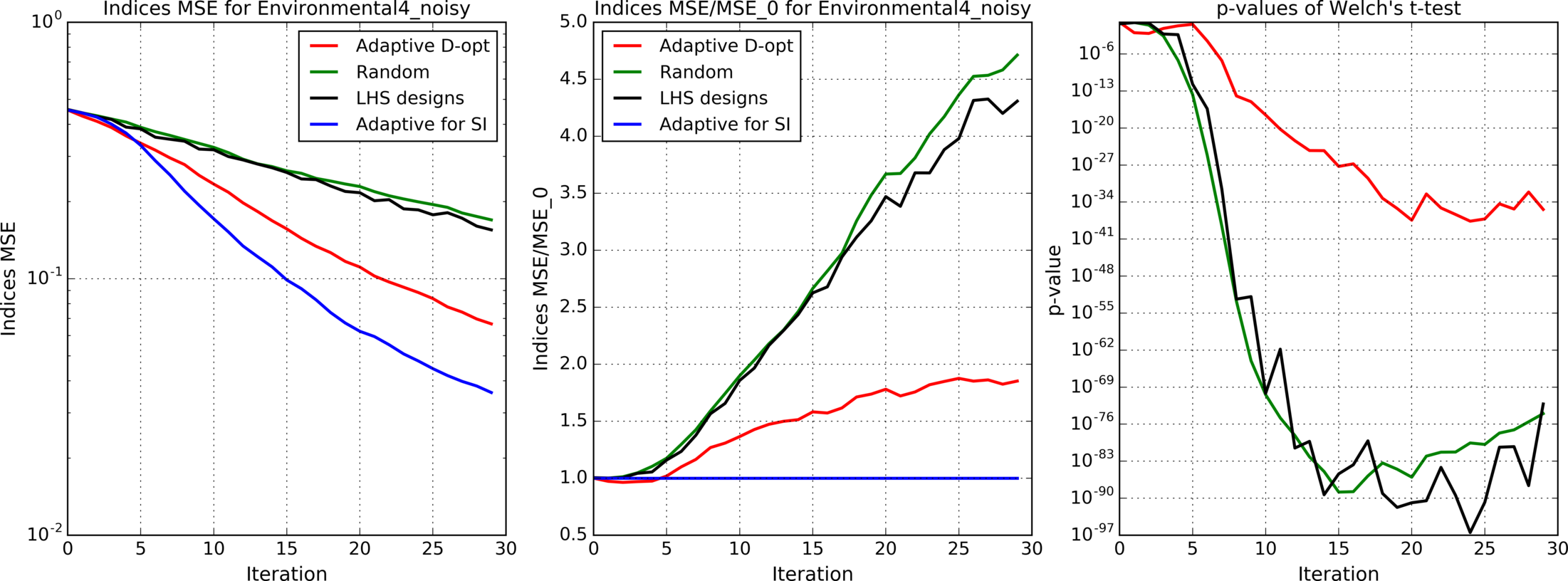}
\caption{\label{pic:Environmental4} Environmental function. $4$-dimensional input.}
\end{figure}

\begin{figure}
    \centering
    \begin{subfigure}[b]{\textwidth}
        \includegraphics[width=\textwidth]{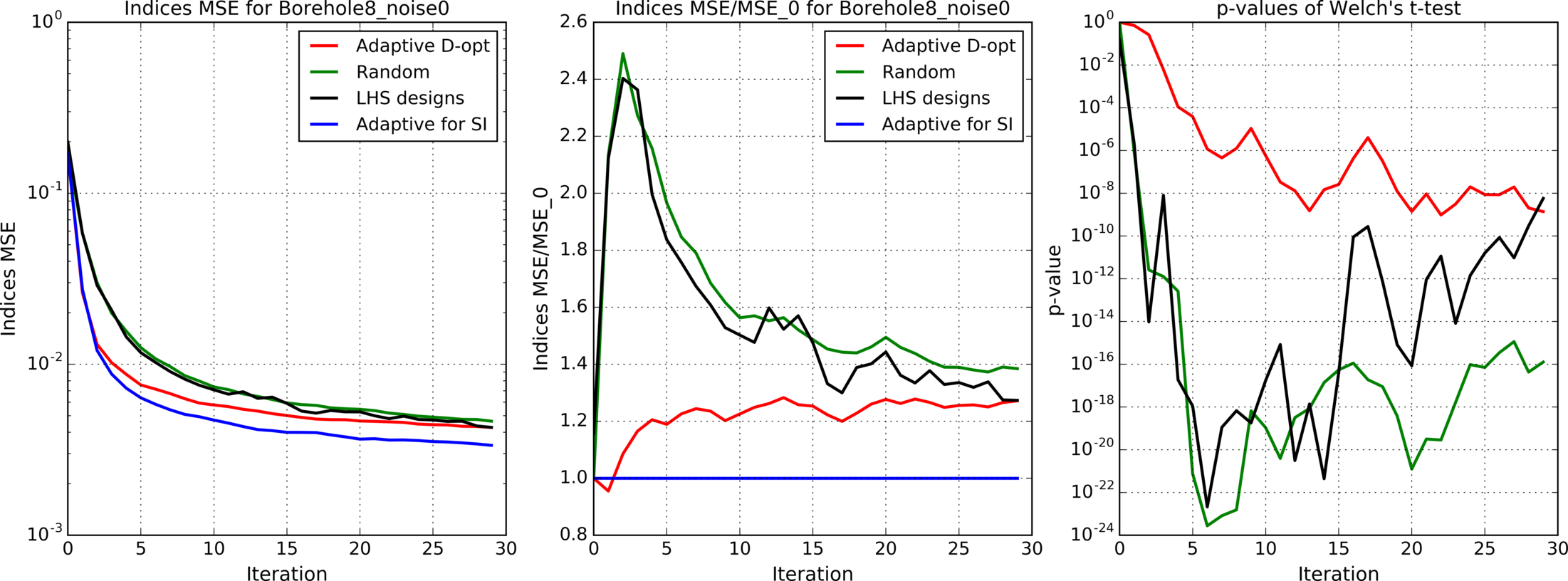}
        \caption{Noise std: 0}
        \label{pic:Borehole8}
    \end{subfigure}

    \begin{subfigure}[b]{\textwidth}
        \includegraphics[width=\textwidth]{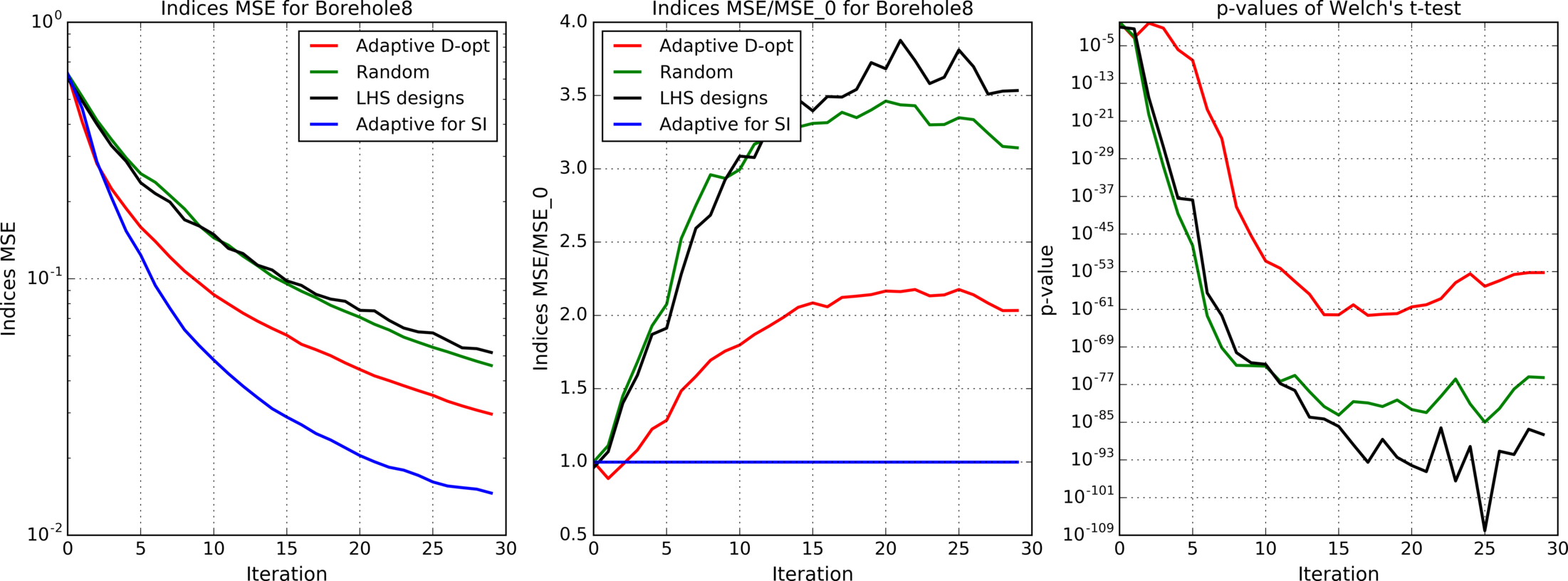}
        \caption{Noise std: 5}
        \label{pic:Borehole8_noise5}
    \end{subfigure}
    \caption{Borehole function. $8$-dimensional input}
\end{figure}

\begin{figure}
    \centering
    \begin{subfigure}[b]{\textwidth}
        \includegraphics[width=\textwidth]{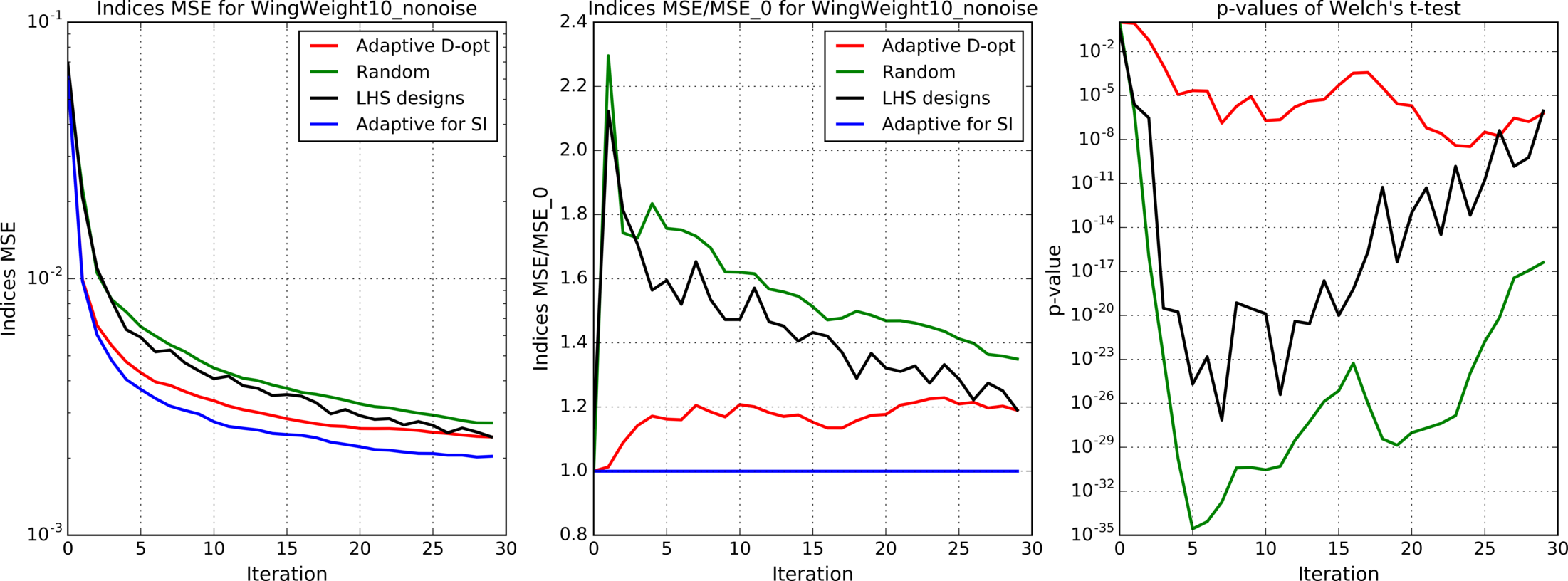}
        \caption{Noise std: 0}
        \label{pic:WingWeight10}
    \end{subfigure}

    \begin{subfigure}[b]{\textwidth}
        \includegraphics[width=\textwidth]{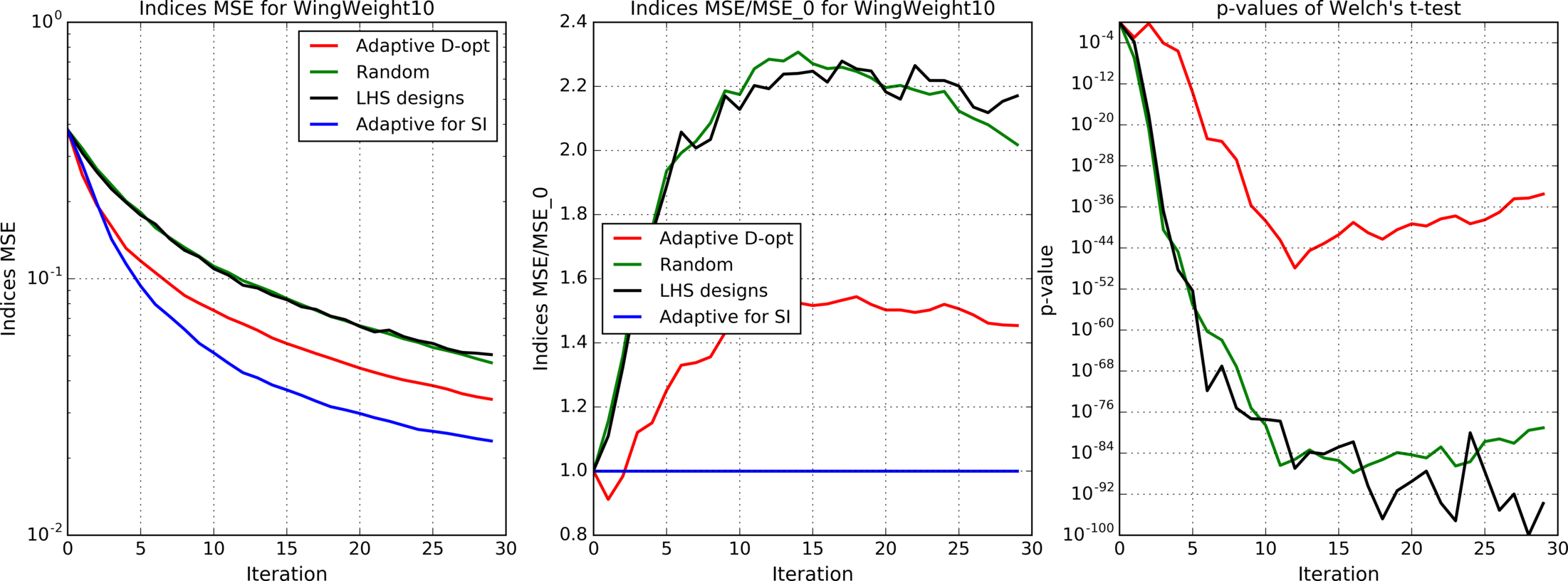}
        \caption{Noise std: 5}
        \label{pic:WingWeight10_noise5}
    \end{subfigure}
    \caption{ WingWeight function. $10$-dimensional input}
\end{figure}


\begin{figure}
\centering
\includegraphics[width=\textwidth]{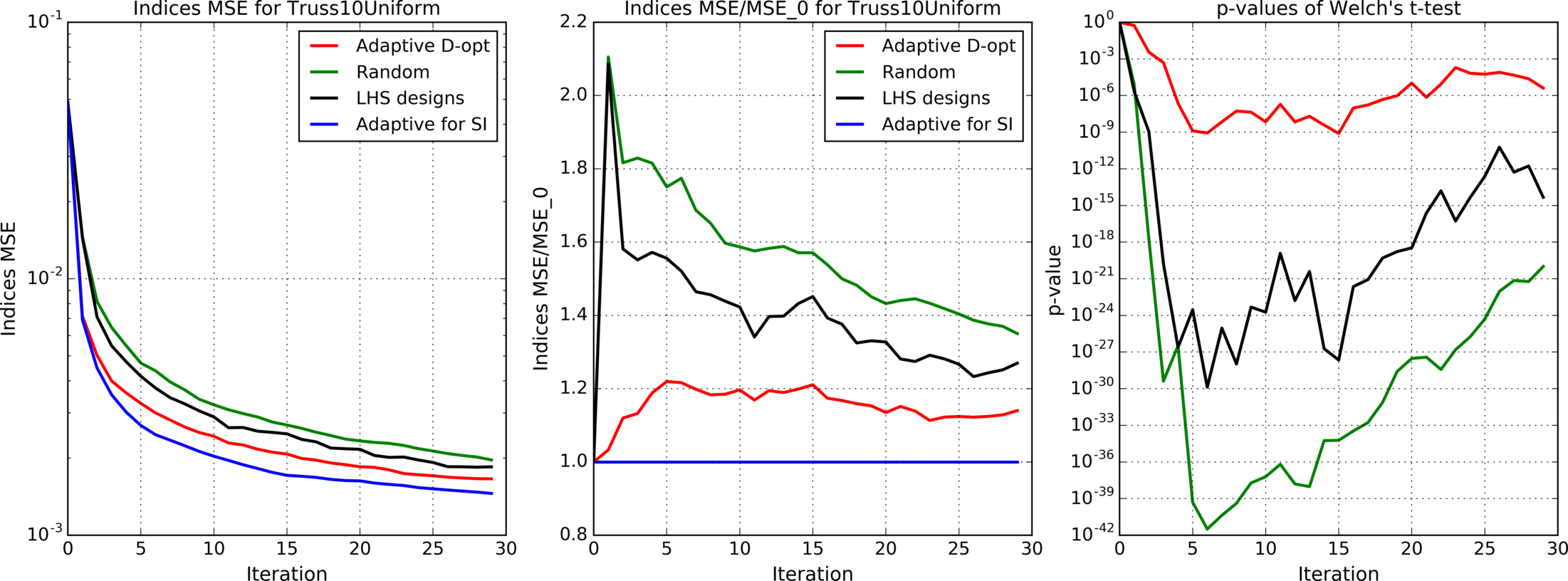}
\caption{\label{pic:Truss10Uniform} Truss model. $10$-dimensional input}
\end{figure}

\begin{figure}
\centering
\includegraphics[width=\textwidth]{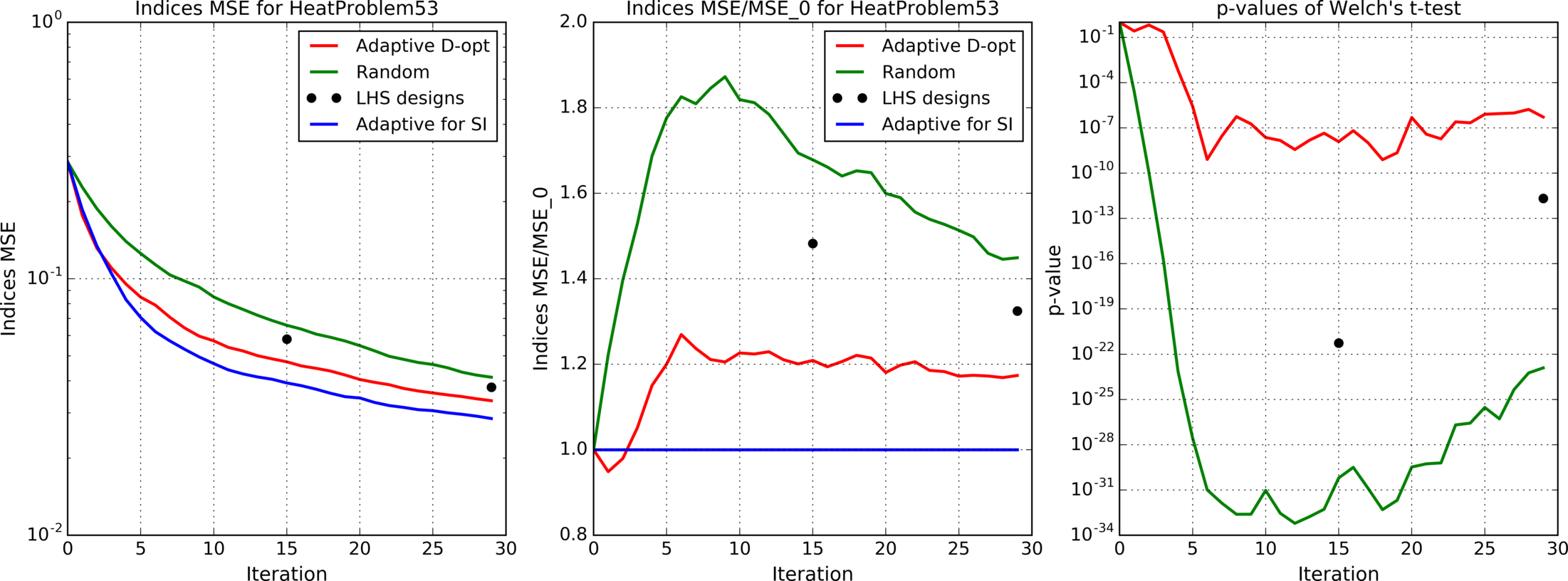}
\caption{\label{pic:HeatProblem53} Heat transfer model. $53$-dimensional input}
\end{figure}

\section{Conclusions}

We proposed the design of experiments algorithm for evaluation of Sobol' indices from PCE metamodel. The method does not depend on a particular form of orthonormal polynomials in PCE. It can be used for the case of different distributions of input parameters, defining the analyzed computational models.

The main idea of the method comes from metamodeling approach. We assume that the computational model is close to its approximating PCE metamodel and exploit knowledge of a metamodel structure. This allows us to improve the evaluation accuracy. All comes with a price: if additional assumptions on the computational model to provide good performance are not satisfied, one may expect accuracy degradation. Fortunately, in practice, we can control approximation quality during design construction and detect that we have selected inappropriate model. Note that from a theoretical point of view, our asymptotic considerations (w.r.t. the training sample size) simplify the problem of accuracy evaluation for the estimated indices.

Our experiments demonstrate: if PCE specification defined by the truncation scheme is appropriate for the given computational model and the size of the training sample is sufficiently large, then the proposed method performs better in comparison with standard approaches for design construction. 
  
\phantomsection

\end{document}